\documentclass[12pt, oneside]{article}
\usepackage[utf8]{inputenc}
\usepackage[english]{babel}
\usepackage{amsfonts}
\usepackage{amsmath}
\usepackage{amsthm}
\usepackage{MnSymbol}
\usepackage{wasysym}
\usepackage{hyperref}
\usepackage{mathtools}
\usepackage{multicol}
\usepackage[letterpaper,margin=2.5cm,left=3cm,right=2.5cm, includehead]{geometry}
\usepackage{graphicx}
\usepackage{float}
\usepackage{enumitem}
\usepackage{setspace}
\usepackage{mathrsfs}
\usepackage{pdfpages}
\usepackage{titling}
\usepackage{tocbibind}
\usepackage[spanish]{babelbib}
\usepackage{authblk}

\usepackage{eso-pic}

\newcommand{\barX}{\overline{X}}

\newcommand{\bartau}{\overline{\tau}}
\newcommand{\bfi}{\mathbf{i}}

\theoremstyle{definition}

\newtheorem{theorem}{Theorem}[section]
\newtheorem{remark}[theorem]{Remark}
\newtheorem{lemma}[theorem]{Lemma}
\newtheorem{proposition}[theorem]{Proposition}
\newtheorem{corollary}[theorem]{Corollary}
\newtheorem{definition}[theorem]{Definition}
\newtheorem{example}[theorem]{Example}

\date{}
\begin{document}

\title{The c-completion of Lorentzian metric spaces}
\author[1]{Saúl Burgos}
\author[2]{José L. Flores}
\author[3]{Jónatan Herrera}
\affil[1]{\small{{\textit{Departamento de Geometría y Topología, Facultad de Ciencias $\&$ IMAG (Centro de Excelencia María de Maeztu)\\
				Universidad de Granada, 18071 Granada, Spain.}}}}
\affil[2]{\small{\textit{Departamento de \'Algebra, Geometr\'{i}a y Topolog\'{i}a\\ Facultad de Ciencias, Universidad de M\'alaga\\ Campus Teatinos, 29071 M\'alaga, Spain.}}}
\affil[3]{\small{\textit{Departamento de Matemáticas, Edificio Albert Einstein, Universidad de Córdoba\\ Campus de Rabanales, 14071 Córdoba, Spain.}}}



\maketitle

\begin{abstract}
   Inspired by some Lorentzian versions of the notion of metric and length
   space introduced by Kunzinger and Sämman \cite{KSlls}, and more recently,
   by Müller \cite{Muller}, and Minguzzi and Sühr \cite{lorentzmetric}, we
   revisit the notion of Lorentzian metric space in order to later construct
   the c-completion of these general objects.  We not only prove that this
   construction is feasible in great generality for these objects, including
   spacetimes of low regularity, but also endow the c-completion with a
   structure of Lorentzian metric space by itself. We also prove that the
   c-completion constitutes a well-suited extension of the original space,
   which really completes it in a precise sense and becomes sensible to
   certain causal properties of that space.
\end{abstract}

\section{Introduction}

   For several decades there has been interest in developing a low-regularity
   approach to General Relativity. It was in 1967 when Kronheimer and Penrose
   \cite{KronPen67} introduced the concept of {\em causal set}, i.e. a 4-tuple
   of the form $(X,\ll,\leq,\rightarrow)$ where $X$ is a set and
   $\ll,\leq,\rightarrow$ are reltions in $X$ playing the role of the
   classical chronological, causal and horismotic relations for spacetimes,
   respectively.

   A similar approach was considered by Bombelli, Lee, Meyer and Sorkin in
   \cite{BomLeeMeySor87}. In that work, motivated by quantum field theory and
   the existence of singularities in spacetimes, the authors introduce a
   notion of causal set with significant differences with respect to the one
   introduced by Kronheimer and Penrose. In fact, this notion of causal set
   only considers an ordering relation $\prec$, that plays the role (in a
   generic way) of the causal relation. In addition, and guided by the idea
   that the models of interest are discrete, they include the finite charater
   of the cardinality of the causal diamonds among their hypotheses.

   Years later, in the context of the future causal boundary of spacetimes,
   Harris \cite{harrisuniversality} introduces the notion of chronological
   set. Similar to the approach by Bombelli et al, Harris' definition only
   requires a set and an ordering relation. However, given that the
   construction of the future causal boundary requires the chronological
   relation, the ordering relation does not coincide with the previous one.
   Harris' work was later extended by one of the authors in
   \cite{flores:revisited}, in the context of the total causal boundary, that
   is, when one considers the future and past boundaries simultaneously, and
   their non-trivial relations.

   Even though these approximations are pretty general (they are applicable to
   discrete models), they also become convenient for situations of greater
   regularity. However, as showed in \cite{piotrc0, futurenotopen, ling2020},
   the low-regular causality theory presents a new range of phenomena: there
   exist causal bubbles, the push-up property may not hold, the chronological
   future and past may not be open or even may depend on the class of curves
   used to define them.

   Recently, the interest in low-regularity approximations to General
   Relativity have received a strong boost thanks in part to the detection of
   gravitational waves \cite{LIGO16} and the existence of more precise
   observations of black holes \cite{EventHor19}. In addition, alternative
   versions of the singularity theorems have been obtained for $C^1$ and
   $C^{1,1}$ spacetimes \cite{kunstein2015, graf2018, graf2020}. This new
   impulse does not seek an approximation of discrete nature (as the ones
   mentioned before), on the contrary it looks for an approximation with a
   continuous structure. The starting point of this new boost is found in the
   works by Kunzinger and S\"amann \cite{KSlls}. Inspired by the classical
   theory of Riemannian length spaces, in these papers the authors define the
   concept of {\em Lorentzian length space} in a precise way. Their approach
   allows them to re-obtain several definitions and results from the theory of
   (smooth) spacetimes, like several steps of the causal ladder and their
   logical relations (the remaning steps were studied in \cite{ACS}).

   Since that work, different authors have made further contributions in the
   field of low-regularity Lorentzian geometry. For instance, M\"uller in
   \cite{Muller,Muller2} have introduced several categories of synthetic
   objects related with Lorentzian length spaces. Concretely, he has
   introduced the notions of \emph{almost pre-length space} and \emph{ordered
   measure space}. The former is defined by a set $X$ endowed with a function
   $\tau$ playing the role of a Lorentzian distance. The function $\tau$ is
   also assumed to be lower-semicontinuous under some appropriate topology.
   The latter is given by a set $X$ endowed with a relation $\leq$, which
   plays the role of the causal relation, joined to a measure ${\rm Vol}$. In
   his works, he has analyzed the relation between both definitions from a
   categorical perspective. He has also defined, following Noldus' ideas in
   \cite{Noldus,Noldus2}, an analog of the Gromov-Hausdorff distance, being
   able to obtain in such categories some pre-compactness results under mild
   hypotheses.

   In parallel, but independently, Minguzzi and Suhr have also introduced a
   notion for Gromov-Hausdorff distance \cite{lorentzmetric}. In their
   approach, they consider the notion of \emph{bounded Lorentzian metric
   space}. This object can be seen as a subcategory of (almost) pre-length
   spaces, since it also includes a set $X$ endowed with a function $\tau$.
   However, these authors impose some further restrictions. For instance, they
   consider a more restricted topology, since the function $\tau$ is required
   to be continuous. They also require a compactness condition for certain
   sets defined in terms of $\tau$. As in M\"uller's works, they also obtain
   (pre-)compactness results for families of bounded-Lorentzian metric spaces
   satisfying mild conditions.

   A promising tool in the context of non-regular spaces is the  causal
   boundary construction. Recall, for instance, the central role played by
   this construction to refine notions like the null infinite or the black
   hole in the regular case \cite{CosFloHer}. These definitions rely on the
   notion of completeness of lightlike geodesics, which can be translated to
   the non-regular context by using the Lorentzian distance (see
   \cite{GKSinext}). Therefore, the extension of the Lorentzian distance to
   the causal boundary (or, more specifically, the causal completion) becomes
   a promising starting point to extend these definitions to non-regular
   spaces.

   A first step in the study of the causal boundary for pre-length spaces has
   been recently given by one of the authors in collaboration with Aké and
   Sol\'is \cite{ABS}. Although their work becomes useful in many situations,
   it is far from being totally general. In fact, they are involved just with
   the partial (future) causal boundary construction, and consequently, they
   are forced to restrict their attention to {\em globally hyperbolic}
   Lorentzian pre-length spaces.

   However, in order to get the most from the causal boundary construction in
   the context of non-regular (metric) spaces, it becomes essential to
   maintain the maximum generality in the process, at the same time that one
   endows the resulting construction with enough (metric) structure. The
   spirit of this paper is based on the pursuit of these competing goals.

   \smallskip

   More precisely, our aim here is threefold. First, we revisit the notion of
   \emph{Lorentzian metric space}. Our approach closely follows the notion of
   (almost) pre-length space, but without considering a concrete topology. In
   this way, the notion of Lorentzian metric space becomes general enough to
   include, not only almost pre-length spaces, but also bounded Lorentzian
   metric spaces and length spaces.

   Next, we show that, under some mild hypotheses, it is possible to define
   the (total) causal completion for Lorentzian metric spaces. We also prove
   that the Lorentzian distance extends naturally to the causal boundary, and
   consequently, the causal completion also inherits a structure of a
   Lorentzian metric space.

   Finally, we show that the causal completion becomes a well-suited extension
   of the original space. In particular, we prove that the causal boundary
   actually ``completes'' the original space and satisfies a sort of ``density
   property'' between points of original space and the boundary of any other
   extension (including those defined by Kunzinger and S\"amman in
   \cite{KSlls}). In addition, we provide some results showing the relations
   between the causal structure of the original Lorentzian metric space and
   its causal completion.

   \smallskip

   The paper is organized as follows. In Section \ref{sec:preliminares} we
   give some preliminaries about Lorentzian length spaces, according to the
   approach developed by Kunzinger and S\"amann \cite{KSlls}. We also recall
   some recent contributions on the subject provided by Müller \cite{Muller},
   and Minguzzi and Suhr \cite{lorentzmetric}.

   Inspired by these recent contributions, in Section
   \ref{sec:lorentzianmetricspaces} we formulate the notion of Lorentzian
   metric space, and deduce its main properties. We also recall some elements
   of the causal ladder needed throughout the manuscript.

   In Section \ref{sec:extensions} we adapt the notion of extension for
   Lorentzian length spaces to Lorentzian metric spaces. We also formalize
   what we understand by {\em completion} in this context.

   In Section \ref{sec:c-completion} we construct the c-completion of a
   Lorentzian metric space, showing that it inherits a structure of Lorentzian
   metric space. Moreover, we prove that this boundary construction verifies
   multiple satisfactory properties, like completeness and density with
   respect to other extensions.

   Finally, in Section \ref{sec:applications} we provide some results showing
   the relations between the causal structure of the original Lorentzian
   metric space and its causal completion. In particular, we obtain
   characterizations of global hyperbolicity and simple causality of a
   Lorentzian metric space in terms of its completion. We also infer
   information about the causal structure of the causal completion.

 \section{Preliminaries}\label{sec:preliminares}


This section is devoted to briefly summarize some central notions and basic definitions from the theory of Lorentzian length/metric spaces \cite{GKSinext,lorentzmetric}.

\smallskip

We begin with some fundamental concepts introduced by Kunzinger and S{\"a}mann in \cite{GKSinext} (see also \cite{ACS,KSlls}). First, we recall the notion of \emph{Lorentzian pre-length space}, where only the Lorentzian distance is considered (in addition to the causal and chronological relations), and then, we provide the full notion of \emph{Lorentzian length space}, which requires a number of mild conditions and includes in particular the notion of length for a suitable class of curves. 


\begin{definition}\label{def:KSlpls}
A {\em Lorenztian pre-length space} is a quintuple $(X,d,\ll, \leq, \tau)$ where
\begin{enumerate}[label=(\arabic*)]
\item $(X,d)$ is a metric space,
\item $\leq$ is a pre-order in $X$ (i.e., it is reflexive and transitive),
\item $\ll$ is a transitive relation contained in $\leq$.
\item $\tau: X \times X \to [0,\infty]$ is a lower-semicontinuous function satisfying
	\begin{itemize}
	\item (Positivity) $\tau (x,y) > 0$ if and only if $x\ll y$,
	\item (Reverse triangle inequality) $\tau (x,z) \geq \tau (x,y) + \tau (x,z)$ 	for all $x \leq y \leq z$.
	\end{itemize}
\end{enumerate}
\end{definition}
\noindent Note that, in contraposition to the theory of spacetimes, the chronological and causal relations, as well as the Lorentzian distance, are introduced axiomatically, preceding the notions of causal curves and their lengths, which are introduced below.

\begin{definition}\label{def:KScurves}
Let $(X,d, \ll, \leq, \tau)$ be a Lorentzian pre-length space. A {\em future directed causal} ({\em timelike}) curve is a non-constant locally Lipschitz map $\gamma : I (\subset \mathbb{R}) \to X(\equiv (X,d))$ with the property that $\gamma (s) \leq \gamma (t)$ ($\gamma (s) \ll \gamma (t)$) for all $s < t$. A future directed causal curve is called {\em null} if there are no points on the curve related with respect to $\ll$. {\em Past directed causal (timelike, null)} curves are defined analogously.
\end{definition}
\noindent This definition, even if very natural, does not recover (in general) the classical notions of causal and timelike curves for spacetimes. In fact, consider for instance the Lorentzian cylinder $(\mathbb{S}^1\times \mathbb{R}, -d\theta^2+dx^2)$. This spacetime is totally vicious, since any pair of points are chronologically related. Consequently, any non-constant locally Lipschitz map $\gamma:I\rightarrow X$ is a timelike curve according to previous definition. However, the curve $\gamma$, even if smooth, does not necessarily satisfy the classical condition of being timelike its velocity $\dot{\gamma}$.
%

This pathology does not constitute, however, an obstacle to assign a Lorentzian length to these curves.
\begin{definition} The $\tau$-{\em length} of a causal curve $\gamma:[a,b]\rightarrow X$ is given by
$$L_{\tau} (\gamma):= {\rm sup} \left\lbrace \sum_{i=0}^{N-1} \tau (\gamma (t_i) , \gamma (t_{i+1})) \right\rbrace,$$
where the supremum runs over all partitions $a = t_0 < t_1 < \cdots < t_{N-1} < t_N = b$ of $[a,b]$.
\end{definition}

There are many other classical notions that can be also defined in the context of pre-length spaces (chronological/causal future/past of a point/set, indecomposable past/future sets, common future/past of a set...). We postpone these definitions to Section \ref{qaz} (see also \cite[Chapter 14]{oneill} for classical definitions of such notions).

The relation $\ll$ allows us to define two natural topologies in terms of the future and past of points of $X$: (i) the \emph{coarsely extended Alexandrov topology}, or \emph{CEAT} for short, defined as the one generated by the subbasis of chronological futures and pasts of points $\left\{ I^{\pm}(x):\; x\in X  \right\}$ (compare with the \emph{chronological topology} in Definition \ref{def:chrtop}); (ii) the \emph{Alexandrov topology}, defined as the one generated by the subbasis of chronological diamonds $\left\{I(x,y):=I^+(x)\cap I^-(y):\; x,y\in X  \right\}$. The CEAT topology is, in general, finer than the Alexandrov one, but they both coincide if $(X,d,\ll,\leq,\tau)$ is strongly causal (see Definition \ref{def:causalityconditions}).

\smallskip

In order to establish the full definition of length space (according to \cite{KSlls}), we still need a number of preliminar notions.
\begin{definition}
\label{def:main:1}
Let $(X,d,\ll,\leq,\tau)$ be a Lorentzian pre-length space, and take $U\subset X$. We say that $U$ is:
\begin{enumerate}[label=(\roman*)]
\item {\em causally (timelike) path-connected} if for any pair of points $x,y\in U$ with $x\leq y$, there exists a causal (timelike) curve $\gamma$ in $U$ joining them. If $U=X$, we just say that the pre-length space is {\em causally path-connected};
\item {\em causally closed} if given two sequences $\left\{ p_n \right\},\left\{ q_n \right\}\subset U$ with $p_n\leq q_n$, $p_n\rightarrow p$ and $q_n\rightarrow q$, it follows that $p\leq q$. If for every point $x\in X$ there exists a neighbourhood $U\subset X$ of $x$ which is causally closed, we say that the pre-length space is {\em locally causally closed}.
\end{enumerate}
\end{definition}
\noindent We say that a Lorentzian pre-length space is \textit{localizable} if every point has a neighbourhood that ``resembles'' a convex neighbourhood. For further details on this definition we refer to \cite[Definition 3.16]{KSlls}.

\begin{definition}\label{def:KSlls}
A {\em Lorentzian length space} is a causally path connected, locally causally closed and localizable Lorentzian pre-length space $(X,d,\ll,\leq,\tau)$ such that the Lorentzian distance function $\tau$ satisfies
$$\tau (x,y) = \sup \left\lbrace L_{\tau} (\gamma): \gamma \text{ is a future causal curve from } x \text{ to } y \right\rbrace.$$
\end{definition}

\medskip

A quite different (and more general) synthetic approach to Lorentzian Geometry has been developed recently by M\"uller \cite{Muller} and Minguzzi and Suhr \cite{lorentzmetric}. They remove the \emph{length} requirement from the definition, making a development closer to metric spaces. More precisely, M\"uller defines an \emph{almost pre-length space} (ALL for short) as a pair $(X,\tau)$ given by a set of points $X$ plus an antireflexive function $\tau:X\times X\rightarrow [0,\infty)$ satisfying the reverse triangle inequality. Minguzzi and Suhr's approach refines the notion of ALL by introducing several additional conditions that better correspond to compact causally convex subsets of globally hyperbolic spacetimes and causets. Both notions are essentially included in the following definition: 


\begin{definition}\label{def:bounded} An \emph{almost pre-length space ALL} is a pair $(X, \tau)$ given by a set $X$ endowed with a map $\tau:X\times X \rightarrow [0, \infty)$ such that:
	\begin{enumerate}[label=(\roman*)]
\item for any $x,y,z\in X$ with $\tau(x,y)>0$, $\tau(y,z)>0$, it follows
  $\tau(x,z)\geq \tau(x,y) + \tau(y,z)$.
\end{enumerate}
An ALL is a \emph{bounded Lorentzian-metric space} if, in addition:
\begin{enumerate}[label=(\roman*)]
  \setcounter{enumi}{1}
\item for every $\epsilon>0$ the sets $\{(p, q): \tau(p,q)\geq \epsilon\}$
are compact with respect to the product topology $\sigma\times \sigma$, where $\sigma$ is the coarsest topology on $X$ for which $\tau$ is continuous;
\item \label{def:bounded:3} \label{def:bounded:iii} $\tau$ distinguishes points, i.e. for every pair $x, y$, $x\neq y$ there exists some $z\in X$ such that either $\tau(x,z)\neq \tau(y,z)$ or $\tau(z,x)\neq  \tau(z,y)$;
\item \label{def:bounded:4} there is a countable subset $S\subset X$ such that the point $z$ in (iii) can be found in $S$.
\end{enumerate}
The authors write $p\ll q$ if $\tau(p,q)>0$, and say that $p$ belongs to the chronological past of $q$ (or $q$ belongs to the chronological future of $p$). By the reverse triangle inequality, $p\ll q$ and $q\ll r$ imply $p\ll r$.
\end{definition}

\begin{remark}
\label{rem:main:3}
Some observations are in order:
\begin{enumerate}
\item The definition of ALL in this paper differs slightly from the original one in \cite{Muller}. In fact, neither we require $\tau$ to be antisymmetric nor its codomain to be $\mathbb{R}$. However, defined in this way, bounded metric spaces are a sub-case of ALL.
\item Condition \ref{def:bounded:4} of previous definition was included as part of the definition in a previous version of \cite{lorentzmetric}, but the authors have shown in later developments that \ref{def:bounded:4} can be deduced from the other conditions (see \cite[Proposition 1.10]{lorentzmetric}).
\item In \cite{harrisuniversality} Harris introduces the notion of separability for pre-chronological sets\footnote{I.e. a pair $(X,\ll)$ formed by a set $X$ and a pre-order $\ll$ on $X$.} in the following way: a pre-chronological set $(X,\ll)$ is \emph{separable} if there exists a countable subset $S\subset X$ such that, for any pair $x,y\in X$ with $x\ll y$, there exists a point $s\in S$ with $x\ll s\ll y$.

  Note however that this separability condition is logically independent from condition \ref{def:bounded:3}. In fact, consider the classical example of a non-causal but chronological spacetime given in \cite[Figure 6]{MinguzziSanchez}  (see Figure \ref{fig:2}). This spacetime is separable, but it does not satisfy \ref{def:bounded:3}: any pair of points $x,y$ in the circle satisfy that $\tau(\cdot,x)=\tau(\cdot,y)$ and $\tau(x,\cdot)=\tau(y,\cdot)$. 

  To disprove that condition \ref{def:bounded:3} implies separability, consider the subset $X$ of the $2$-dimensional Minkowski spacetime given by $X=\cup_{i=1,2}\mathbb{Q}\times \left\{ i \right\}$. Let $\tau$ be the Lorentzian distance for Minkowski spacetime, and consider the restriction $\tau|_X$. According to this definition, $\tau_X$ is a Lorentzian distance satisfying both \ref{def:bounded:3} and \ref{def:bounded:4} 
  However, there are no triples $x,y,z$ with $x\ll y\ll z$, and consequently, the separability condition is not satisfied.
\begin{figure}
 \centering
 \includegraphics[scale=0.4]{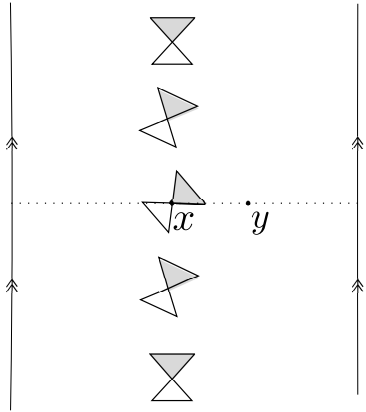}
 \caption{Consider an unbounded rectangle $M=[-1,1]\times \mathbb{R}/\sim$, where  $(-1,y)\sim (1,y)$ for all $y\in \mathbb{R}$, and endow it with a Lorentzian metric whose cones behave as depicted in the figure. The circle $\gamma:[-1,1]\rightarrow M$, $\gamma(t)=(t,0)$, is lightlike. In this spacetime (thus, a separable Lorentzian metric space), we have $I^{\pm}(\gamma(t_1))=I^{\pm}(\gamma(t_2))$ for any pair $t_1,t_2\in [-1,1]$. Therefore, $\tau$ cannot distinguish $\gamma(t_1)$ from $\gamma(t_2)$, and condition \ref{def:bounded:3} in Definition \ref{def:bounded} does not hold.}
 \label{fig:2}
\end{figure}

\end{enumerate}
\end{remark}

\section{Lorentzian Metric Spaces}\label{sec:lorentzianmetricspaces}

%

Our starting point consists of considering the common denominator of the models discussed in previous section in order to formalize a general notion of Lorentzian metric space, which allows us to analyze the minimum requirements necessary to construct the c-compltion of that space. More precesely, we consider a notion of Lorentzian metric space similar to Almost Lorentzian length spaces, but incluing two main differences: (i) we do not make any specific choice about the topology (although in the definition of ALL it does not appear, the author of \cite{Muller} defines a specific topology constructed from $\tau$); (ii) we remove the term \emph{length} from the definition to emphsize that our approach does not depend on the behavior of the curves in the space. Further requirements, like separability, are included later whenever they are strictly necessary.

\begin{definition} \label{def:lorentzianmetricspace}
	A {\em Lorentzian metric space} $(X,\sigma,\tau)$ is a 
        topological set $(X,\sigma)$ endowed with a function $\tau : X \times X \to [0,\infty]$, called the {\em Lorentzian distance function}, which satisfies the following properties:
	\begin{itemize}
		\item[(i)] $\tau$ is lower-semicontinuous with respect to $\sigma$;
		\item [(ii)] the reverse inequality holds, i.e.
			$$\tau (x,z) \geq \tau (x,y) + \tau (y,z)\qquad\hbox{if $\tau (x,y), \tau (y,z)>0$}.$$
		\end{itemize}
	In this setting two points $x,y\in X$ are said {\em chronologically related}, $x\ll y$, iff $\tau(x,y)>0$.
              \end{definition}
          Note that, in contraposition to pre-length spaces, we do not require the topology of the Lorentzian metric space to be metrizable. In fact, this condition is only useful when length of curves are considered, something which is not required for our study. Moreover, we believe that the metrizability of the topology associated to classical metric spaces is a peculiarity of the ``Riemannian character'' of these spaces, which is not justified a priori in the Lorentzian setting. On the other hand, our definition is very close to the notion of almost pre-length space. The only difference is that here a generic topology is axiomatically incorporated from the beginning, not derived from anything. There are two reasons for this: (1) it permits to formulate the lower semicontinuity property, (2) it allows us to choose the topology that best suits to each situation (see for instance the analysis in \cite{costa:addendum20}).

          

\begin{definition}\label{xuxu} Let $(X,\sigma,\tau)$ be a Lorentzian metric space. A {\em causal relation compatible with} $(X,\sigma,\tau)$ is a pre-order relation $\leq$ on $X$ such that 
\begin{itemize}	
\item[(i)]$x\leq y$ if $\tau(x,y)>0$ ($\leq$ {\em contains or extends} $\ll$), and 
\item[(ii)] $\tau (x,z) \geq \tau (x,y) + \tau (y,z)$ if $x\leq y\leq z$ ({\em reverse triangle inequality of} $\tau$ for $\leq$).
 \end{itemize}
\end{definition}


Given a Lorentzian metric space $(X,\sigma,\tau)$, there are many possibilities for $\leq$. The following causal relations are distinguished between them for being the ``minimum'' and ``maximum'' ones, respectively:
\begin{itemize}
	\item $x\leq_{0} y$ iff either $x=y$ or $\tau(x,y)>0$;
	\item $x \leq_{\tau} y$ iff $I^-(x)\subset I^-(y)$ and $I^+(y)\subset I^+(x)$ (firstly introduced in \cite{MinguzziSanchez}).
\end{itemize}  
We believe however that a canonical choice for $\leq$ with satisfactory results in all the cases does not exist. The eventual flaws of a generic choice for $\leq$ are particularly evident if we choose $\leq_{0}$, but they are also present for $\leq_{\tau}$. In fact,
consider the region $M=\mathbb{L}^2\setminus \left\{ (1,1) \right\}$ of Minkowski plane $\mathbb{L}^2$, which contains the points $(0,0)$ and $(2,2)$. These points are causally related in $\mathbb{L}^2$, since they can be connected by a lightlike curve. Moreover, the inclusions $I^+((2,2))\subset I^+((0,0))$ and $I^-((0,0))\subset I^-((2,2))$ clearly hold. However, these points are no longer related in $M$, even though previous inclusions still hold. Consequently, and following the notion of pre-length space of previous section, we renounce to specify a generic choice for $\leq$, and leave to each one the specific choice which better adapts to each case. Finally, in order to emphasize the points of our work where $\leq$ is actually needed, we have opted for getting the causal relation out of Definition \ref{def:lorentzianmetricspace}.

\begin{remark}
	There is a natural way to incorporate $\leq$ into Definition \ref{def:lorentzianmetricspace} via the domain of $\tau$, thus, without including additional objects. Instead of assuming that the domain of $\tau$ is $X\times X$, one can suppose that $\tau$ is defined on some subset $J\subset X\times X$ verifying: 
		\[
		(x,y), (y,z)\in J\;\Rightarrow \; (x,z)\in J.
		\] 
		Then, we write $x\leq y$ iff $(x,y)\in J$ (the relation $\ll$ remains defined as before).
		Note that the Lorentzian distance between pairs $(x,y)\not\in J$ is no longer defined according to this approach. (See \cite{McCann} for an alternative, but similar, way to incorporate $\leq$ to $\tau$).
\end{remark}

%
%

Next, we recall the following result, which analyzes the role of the lower-semicontinuity of $\tau$ and its relation with the topology $\sigma$ and the corresponding chronological relation (see \cite[Lemma 2.12]{KSlls}).


\begin{proposition}
\label{thm:main:1}
Let $(X,\sigma,\tau)$ be a Lorentzian metric space. If $x\ll y$ then there exists open sets $U,V$ with $x\in U,y\in V$ such that $p\ll q$ for any $p\in U$, $q\in V$. In particular, the space topology $\sigma$ is finer than the CEAT topology. 
\end{proposition}
\begin{proof}
Since $x\ll y$, we have by definition that $\tau(x,y)>0$. Then, by the lower semi-continuity of $\tau$, we can find open sets $U,V$ of $x$ and $y$ resp., such that  any pair of points $p\in U$ and $q\in V$ satisfies $\tau(p,q)\geq \tau(x,y)/2>0$. Hence, $p\ll q$, as desired.
\end{proof}




A key property which combines the chronological and causal relations is the so-called {\em push-up property}, i.e.,  for any $x,y,z\in X$ satisfying $x\leq y\ll z$ or $x\ll y\leq z$, necessarily $x\ll z$ (compare with \cite[Lemma 2.10]{KSlls}). The following result clarifies the logical relation between this property and the reverse triangle inequality (Definition \ref{xuxu} (ii)). But, first, let us introduce a definition:
\begin{definition}
	\label{def:main:6}
	A Lorentzian metric space $(X,\sigma,\tau)$ is \emph{chronologically dense} if, for every point $x\in X$ with $I^-(x)\neq \emptyset$ (resp. $I^+(x)\neq \emptyset$),   
	there exists a future (resp. past) chain $\varsigma$ with topological limit $x$. 
\end{definition}
\begin{proposition}
\label{thm:main:5}
Let $(X,\sigma,\tau)$ be a 
Lorentzian metric space, and denote by $\leq$ a pre-order relation in $X$ containing $\ll$. If the reverse triangle inequality of $\tau$ is satisfied for $\leq$, then the push-up property is also satisfied. The converse holds if $(X,\sigma,\tau)$ is chronologically dense. In particular, under this last hypothesis, condition (ii) in Definition \ref{xuxu} can be replaced by the push-up property. 
\end{proposition}
\begin{proof}

Assume that $\tau (x,z) \geq \tau (x,y) + \tau (y,z)$ whenever $x \leq y \leq z$. If $x,y,z \in X$ satisfy $x \ll y \leq z$ (the other case is analogous) then 
  \[ \tau (x,z) \geq \tau (x,y) + \tau (y,z) > 0,
  \]
and consequently $x \ll z$.

  \smallskip
  
Assume now that the push-up property and the chronological density hold. Observe that the reverse triangle inequality is satisfied for $\tau$ whenever the three points are chronologically related. So, we can assume that at least one of the causal relations $x\leq y\leq z$ is not a chronological relation. Let us assume that, for instance, $x\leq y$ but $x\not\ll y$ (the rest of cases are analogous). Then, it is clear that $\tau(x,y)=0$.

If $y\leq z$ but $y\not\ll z$, we also have $\tau(y,z)=0$, and the inequality follows trivially. So, the remaining possibility is $y\ll z$, i.e., $\tau(y,z)>0$. In this case, the push-up property yields $\tau(x,z)>0$. In order to prove that $\tau(y,z)\leq \tau(x,z)$ let us consider $\varsigma=\left\{ y_n \right\}\subset I^+(y)$ a past chronological chain with topological limit $y$. The existence of such a sequence is ensured by the chronological density and the fact that $z\in I^{+}(y)$. Observe that, in the one hand, the push-up property ensures that $x\ll y_n$ for all $n$. On the other hand, by recalling that $y\ll z$ and Proposition \ref{thm:main:1}, it follows that for $n$ big enough $y_n\ll z$. Hence,
\[x\ll y_n\ll z
\]
and by the reverse triangle inequality,
\[
\tau(x,z)\geq \tau(x,y_n)+\tau(y_n,z).
\]
Now recall that $\tau$ is lower semi-continuous, and so, for any $\epsilon>0$ and $n$ big enough we have that

\[\tau(x,y_n)>\tau(x,y)-\epsilon/2=-\epsilon/2,\qquad \tau(y_n,z)> \tau(y_n,z)>\tau(y,z)-\epsilon/2.
\]
Summing up, for any $\epsilon>0$ it follows that:
\[
\tau(x,z)\geq \tau(x,y_n)+\tau(y_n,z)>\tau(y,z)-\epsilon,
\]
thus $\tau(x,z)\geq \tau(y,z)$ and the result follows.


\end{proof}

\begin{remark}
  The chronological density condition for the converse of previous result is necessary. In fact, let $X=\left\{ x,y,z \right\}$ be a set formed by three points and endowed with the discrete topology. Consider a Lorentzian distance $\tau$ in $X$ whose unique non-zero values are $\tau(x,z)=1$ and $\tau(y,z)=2$, and a causal relation $\leq$ just consisting of $x\leq y\leq z$. Clearly, $\tau$ satisfies the reverse triangle inequality (there are no three points chronologically related) and it is continuous. Even more, the relations $x\leq y\ll z$ and $x\ll z$ hold, and consequently, the push-up property also holds. However, $\tau(x,z)<\tau(y,z)$ with $x\leq y\leq z$, and thus, the reverse triangle inequality is not extended to the causal relation. 
\end{remark}

\begin{definition}
\label{def:fullspace}
A Lorentzian metric space is {\emph{full} (see \cite{Muller})} if the past and future of any of its points are non-empty.
\end{definition}

\medskip

Next, we recall the adaptation to (pre-)length spaces of the causal ladder of spacetimes provided in \cite{ACS} (see also \cite[Sections 2.7 and 3.5]{KSlls}). For convenience, we directly formulate it in the context of Lorentzian metric spaces. 






%
%

\begin{definition}\label{def:causalityconditions}
	A Lorentzian metric space $(X, \sigma, \tau)$ endowed with a compatible causal relation $\leq$ is
	\begin{enumerate}[label=(\alph*)]
		\item\label{esc:1} {\em chronological} if $x \not\ll x$ for any $x \in X$.
		\item\label{esc:2} {\em causal} if $x \leq y$ and $y \leq x$ implies $x = y$.
		\item\label{esc:3} {\em distinguishing} if $I^{-} (x) = I^{-} (y)$ implies $x=y$ (\emph{past-distinguishing)} and $I^{+} (x) = I^{+} (y)$ implies $x=y$ (\emph{future-distinguishing)}.
		\item\label{esc:4} {\em strongly causal} if it is distinguishing
		and the space topology $\sigma$ agrees with the Alexandrov topology.
		\item\label{esc:5} {\em stably causal} if $K^{+}$ is antisymmetric, where $K^{+}$ is the smallest closed and transitive relation containing $J^+:=\left\{ (x,y)\in X\times X: x\leq y \right\}$.
		\item\label{esc:6} {\em causally continuous} if it is distinguishing and reflective, i.e. if $I^{+} (x) \subset I^{+} (y)$ (resp. $I^{-} (y) \subset I^{-} (x)$) implies $I^{-} (y) \subset I^{-} (x)$ (resp. $I^{+} (x) \subset I^{+} (y)$). 
		\item\label{esc:7} {\em causally simple} if it is causal and the sets $J^{\pm} (x)$ are closed for any $x \in X$.
		\item\label{esc:8} {\em globally hyperbolic} if it is causal and the \emph{causal diamonds} $J(x,y) = J^{+} (x) \cap J^{-} (y)$ are compact for any $x,y \in X$.
	\end{enumerate}
\end{definition}
\noindent We postpone to a subsequent paper the proof that these properties are logically related between them (in the sense that each level of the ladder implies the previous one under mild conditions), and the redundancy of the distinguishing condition in (d), as it occurs in the context of pre-length spaces (where additional structure is considered). Let us see however Proposition \ref{krut}.

\smallskip

We conclude this section with the following straightforward result:

\begin{proposition}
\label{prop:main:3}
If $(X,\sigma,\tau)$ is a distinguishing Lorentzian metric space then $\tau$ distinguishes points (according to Definition \ref{def:bounded} \ref{def:bounded:iii}).
\end{proposition}
\begin{proof}
 Assume by contradiction that $\tau$ does not distinguish two points $x,y\in X$. Then, from the definition of $\ll$, and for any $z\in X$, we have $x\ll z$ iff $y\ll z$ and $z\ll x$ iff $z\ll y$. In conclusion, $I^{\pm}(x)=I^{\pm}(y)$, in contradiction with the hypothesis.
\end{proof}


\section{Extensions of Lorentzian metric spaces}\label{sec:extensions}

In this section we are going to adapt the notion of extension of pre-length spaces introduced in \cite[Definition 3.1]{GKSinext} to our setting.

%


\begin{definition}\label{def:extension}
Let $(X,\sigma,\tau)$ be a Lorentzian metric space. We say that a Lorentzian metric space $(\tilde{X},\tilde{\sigma},\tilde{\tau})$ is an {\em extension} of $(X,\sigma,\tau)$ if
\begin{itemize}[label=(\roman*)]
  
\item[(i)] \label{exten:top} there exists a topological embedding $j : (X,\sigma) \to (\tilde{X},\tilde{\sigma})$ such that the image is a proper subset of $\tilde{X}$,

\item[(ii)] \label{exten:dis} $\tilde{\tau}(j(x),j(y)) \geq \tau(x,y)$ for any pair of points $x,y\in X$.
\end{itemize}
If we also have causal relations $\tilde{\leq}$ and $\leq$ compatible with $(X,\sigma,\tau)$ and $(\tilde{X},\tilde{\sigma},\tilde{\tau})$, resp., we say that $\tilde{\leq}$ extends $\leq$ if 
\begin{itemize}
\item[(iii)] $x\leq y$ implies $j(x)\tilde{\leq}j(y)$ for any $x,y\in X$.
\end{itemize} 
%
%

If $(X,\sigma,\tau)$ admits a (proper) extension then it is called {\em extendible}, and it is called {\em inextendible} otherwise. We define the \emph{boundary of the extension} as $\tilde{\partial} X:=\tilde{X}\setminus j(X)$
\end{definition}

\begin{remark} The requirement (i) in previous definition is very natural, since it corresponds to the natural idea of extending the corresponding topological space. However, the condition (ii) needs some explanation. In fact, in contraposition to what happens with the chronology, the Lorentzian distance is not necessarily preserved when passing to the extension. Note however that this is reasonable if we want to include in our approach the spaces whose $\tau$ is expressable as the supremum of lengths of curves joining the corresponding points. In fact, in this case, the additional curves provided by the extension may increase the supremum mentioned above, and consequently, the Lorentzian distance between points of the original space may be strictly bigger.
\end{remark}


%

The notion of extension presented here is sufficiently general to include many constructions, some of them certainly useless; for instance, one can obtain an extension just by attaching appropriately a single point to $X$. In general, some additional properties for the extensions are required in order to obtain some benefit from them. Here, we are specially interested in extensions which yield a ``sort of completion'' for the initial Lorentzian metric space. More precisely, we introduce the following notions: 
\begin{definition}\label{defcompl}
  Let $(\tilde{X},\tilde{\sigma},\tilde{\tau})$ be an extension of a Lorentzian metric space $(X,\sigma,\tau)$. We say that $(\tilde{X},\tilde{\sigma},\tilde{\tau})$ is a \emph{completion} of $(X,\sigma,\tau)$ if
  \begin{enumerate}[label=(\roman*)]
  \item \label{defcompl1} it \emph{preserves the Lorentzian distance}, i.e. $\tilde{\tau}(j(x),j(y))=\tau(x,y)$,
  \item {it is \emph{chronologically complete}, i.e. any chronological chain $\left\{ {\tilde{x}}_n \right\}\subset \tilde{X}$ (i.e. a sequence satisfying either $\tau(\tilde{x}_n,\tilde{x}_{n+1})>0$ for all $n$ or $\tau(\tilde{x}_{n+1},\tilde{x}_{n})>0$ for all $n$), has some endpoint in $\tilde{X}$, and}
  \item it is \emph{tight}, i.e. $j(X)$ is (topologically) dense in $\tilde{X}$.
  \end{enumerate}
\end{definition}

\section{The c-completion of Lorentzian metric spaces}\label{sec:c-completion}

Our main objective is to develop the c-completion construction (initially conceived for spacetimes, but formally introduced for chronological sets) for Lorentzian metric spaces. We will prove that the resulting space is a {\em completion} of the original Lorentzian metric space in the sense defined in previous section, i.e. an {\em extension} of the original Lorentzian metric space that is {\em chronologically complete}, {\em tight} and {\em preserves the Lorentzian distance}. Moreover, we will prove that it is well-suited with respect to other extensions.

\subsection{A review on chronological sets}\label{qaz}

Let $(X,\sigma,\tau)$ be a Lorentzian metric space, and denote by $\ll$ the chronological relation (pre-order) determined by $\tau$. In order to construct the c-completion $\overline{X}$ just as a {\em set of points}, we only need to consider the pair $(X,\ll)$.

We begin by reviewing some classic concepts associated to a {\em pre-chronological sets} (recall footnote 1). We say that a sequence $\left\{ x_n \right\}$ in $X$ is a \emph{future chain} if $x_n\ll x_{n+1}$ for all $n$. A subset $P \subset X$ is called a \emph{past set} if it coincides with its past, i.e., $P = I^{-}[P] = \{  p\in X : p \ll q \text{ for some } q \in P\}$. The \emph{common past} of a subset $S \subset X$ is defined as $\downarrow S := I^{-}[\{ p \in X : p \ll q, \forall q\in S \}]$. A past set that cannot be written as the union of two proper subsets, both of which are also past sets, is called \emph{indecomposable past} set, IP for short. An IP which coincides with the past of some point, i.e. $P = I^{-} (p)$, $p \in X$, is called \emph{proper indecomposable past set}, PIP. Otherwise, it is called \emph{terminal indecomposable set}, TIP. The dual notions of the concepts introduced here (past chain, future set, $\uparrow S$, IF, TIF, PIF) are defined just by changing the roles of past and future.

In order to closer analyze the notion of indecomposability 
we will focus on past sets (for future sets is analogous). In \cite[Theorem 3]{harrisuniversality}, Harris showed, in the context of \emph{chronological sets} 
(i.e., separable pre-chronological sets), that a past set is an IP if and only if it is the past of a future chain. Hence, any IP $P$ in a chronological set can be \emph{generated} by a future chain $\varsigma = \{ p_n\}$, i.e., $P = I^{-}[\varsigma]$. Note however that the separability condition is not enough to ensure that the pasts of points $I^-(x)$ are indecomposable, as we can see from the following example.

\begin{example}
  Consider the subset $X$ of $\mathbb{L}^2$ given by
  \[
    X=\left\{ (0,1) \right\} \cup L_-\cup L_+\subset \mathbb{L}^2,
  \] where $L_{\pm}:=\left\{(\pm 1/2,y):y<0  \right\}$. Denote by $\ll$ the chronology in $X$ inherited from the spacetime $\mathbb{L}^2\setminus \left\{ (0,y):y\leq 0 \right\}$ (in particular, no point of $L_-$ is chronologically related with some point of $L_+$, and viceversa).

It is straightforward to see that  $(X,\ll)$ is separable, since any pair of points $x\ll y$ admits some point $z\in X$ with $x\ll z\ll y$. However, by construction, $I^-((0,1))=L_-\cup L_+$, and each $L_{\pm}$ is an IP ($L_{\pm}=I^-[\varsigma_{\pm}]$, where $\sigma_{\pm}=\left\{ (\pm 1/2,-1/n) \right\}$).
\end{example}

In order to identify points with PIPs we need to assume that  $(X,\sigma,\tau)$ is {chronologically dense} (recall Definition \ref{def:main:6}). 
\begin{lemma}
\label{thm:main:6}
Assume that a Lorentzian metric space $(X,\sigma,\tau)$ is separable and chronologically dense. Then, for any $x\in X$, $I^{\pm}(x)$ in an IP.
\end{lemma}
\begin{proof}
  Let us focus on $I^-(x)$, as the future will be analogous. There are two possibilities: Either $I^-(x)$ in empty (and so, it is an IP), or there exists a future chronological sequence $\varsigma=\left\{ x_n \right\}$ with topological limit $x$. For the latter, take $y\in I^-(x)$ and recall that, from Proposition \ref{thm:main:1}, there exists an open set $x\in U\subset I^+(y)$. But then $y\ll x_n$ for $n$ big enough, and so, $I^-(x)\subset I^-[\varsigma]$. As the other contention follows from transitivity, we deduce that $I^-(x)=I^-[\varsigma]$, and therefore, that $I^-(x)$ is an IP from separability.
\end{proof}

In particular, any separable and chronologically dense Lorentzian metric space satisfies both, that indecomposable past/future sets are past/future of  future/past chains and that the past/future of points are indecomposable. The final result of this section provides sufficent conditions to chronologial density.

\begin{proposition}
\label{thm:main:4}
Any separable strongly causal Lorentzian metric space $(X,\sigma,\tau)$ is chronologically dense. 
\end{proposition}
\begin{proof}
%
First, let us prove that for a given $x\in X$ the set $P=I^-(x)$ is actually a past set. From the transitivity of $\ll$ it is clear that $I^-[P]\subset P$. Hence we need to show the other inclusion. To this aim, observe that $y\ll x$ if $y\in P$. The separability condition ensures the existence of $z\in X$ such that $y\ll z\ll x$. By construction, $z\in P$, and so, $y\in I^-[P]$. Therefore, $I^-(x)$ is a past set. Now, again by using the separability condition, we can construct an IP, denoted by $P$, inside $I^-(x)$ by considering a future chronological chain $\varsigma=\left\{ x_n \right\}$ contained in $I^-(x)$. By Zorn's lemma, we can assume that such an IP is maximal among the IPs contained in $I^{-} (x)$. Our aim now is to prove that $\varsigma$ converges to $x$. 

Let us assume by contradiction that $\varsigma$ does not converge to $x$, and so, that there exists an open set $x\in U$ with $x_n\not\in U$ for all $n$. In such a case we can construct a new future chronological chain in the following way: For $x_1$, let $V_1$ be the open set obtained from Proposition \ref{thm:main:1} so $V_1\subset I^+(x_1)$. As $(X,\sigma,\tau)$ is strongly causal, we can take $x_1^-,x_1^+$ with $x\in I(x_1^-,x_1^-)\subset V_1\cap U$. By construction that $x_1\ll x_1^{-}\ll x$. We can then repeat the process for $x_2$, obtaining $x_2^-,x_2^+$ with $x\in I(x_2^-,x_2^+)\subset V_2\cap I(x_1^-,x_1^+)$. Observe that from construction, $x_2,x_1^-\in I^{-}(x_2^-)$.

Therefore, we can obtain a future chronological chain $\varsigma^-=\left\{ x_n^- \right\}\subset U$ and with $x_n\ll x_n^-$ for all $n$. But then $P \subsetneq I^-[\varsigma^-]$, being the latter an IP, in contradiction with the maximality of $P$. In conclusion, $\varsigma$ converges to $x$.

Finally, we are going to show that $P=I^-(x)$. For this, take $y\ll x$ and let us show that $y\in P$. Again from Proposition \ref{thm:main:1} we can find $x^-,x^+$ so $x\in I(x^-,x^+)\subset I^+(y)$ (here, we are using again strong causality). As $\varsigma$ converges to $x$, $x_n\in I(x^-,x^+)$ for $n$ big enough, and so, $y\ll x^-\ll x_n$. Therefore, $y\in I^-[\varsigma]$, as desired.
\end{proof}



\subsection{On the construction of the causal completion}

From previous discussion, in order to ensure a satisfactory identification between points and indecomposable past/future sets, we can require that $(X,\sigma,\tau)$ satisfies both, separability and chronological density (or separability and strong causality, according to Proposition \ref{thm:main:4}). As a first step towards the construction of the c-completion, we need to introduce the notions of \emph{future and past} completions.

The \emph{future chronological completion}, denoted by $\hat{X}$, is the set of all IPs. 
Since $(X,\sigma,\tau)$ is chronologically dense, we can define a map $x \mapsto I^{-} (x)$ from $X$ to $\hat{X}$, which is injective if $(X,\sigma,\tau)$ is past-distinguishing. Consequently, in this case, we can write $\hat{X} = X \cup \hat{\partial} X$, where $\hat{\partial}X$ denotes the set of all TIPs of $X$, which is called the \emph{future chronological boundary} of $X$. In analogous fashion, we define the notion of \emph{past completion} and \emph{boundary} by considering IFs, PIFs and TIFs. For a deeper analysis on the future chronological completion we refer to \cite{harrisuniversality}. 

 In order to define a \emph{(total)} c-completion, we need to glue together both completions by establishinging some relations between points of both constructions. In fact, it is clear that any PIP and PIF associated to the same point must be identified, but further relations need also to be considered (see \cite{szabados1}). To this aim, we introduce the following definition:


\begin{definition}\label{srelation}
	If $P$ is a maximal IP contained in $\downarrow F$ and $F$ is a maximal IF contained in $\uparrow P$ then we say that $P,F$ are \textit{S-related}, and denote it by $P \sim_{S} F$. Furthermore, we say that $P \sim_{S} \emptyset$ (resp. $ \emptyset \sim_{S} F$) if $P$ (resp. $F$) is a (non-empty) indecomposable past (resp. future) set and it is not S-related to any other indecomposable set.
      \end{definition}

      \noindent Consider the product $\hat{X}_{\emptyset}\times \check{X}_{\emptyset}$, where $\hat{X}_{\emptyset} := \hat{X} \cup \{\emptyset\}$ and $\check{X}_{\emptyset} := \check{X} \cup \{\emptyset\}$, and take the set formed by the pairs $(P,F)\in \hat{X}_{\emptyset}\times \check{X}_{\emptyset}$ which are $S$-related according to previous definition. In order to identify $X$ with a subset of the product $\hat{X}_{\emptyset}\times \check{X}_{\emptyset}$ via the map $x\mapsto (I^-(x),I^+(x))$, there are several difficulties to face. On the one hand, if the space is not past-distinguishing, two different points $x,y\in X$ with $I^-(x)=I^-(y)$ could be identified with a singular point in $\overline{X}$ (the same happens if the space is not future-distinguishing). On the other hand, the relation $I^-(x)\sim_S I^+(x)$ is not guaranteed in general. Of course, these problems are solved for strongly causal spacetimes (see \cite[Prop. 5.1]{szabados1}), but to be finer we need some additional work.

\begin{definition}
\label{def:main:2}
We say that $(X,\ll)$ satisfies the \textit{S-property} if $I^-(x)\sim_SI^+(x)$ for any $x\in X$, and there is no other PIF $F$ (resp. PIP $P$) $S$-related with $I^-(x)$ (resp. $I^+(x)$).
\end{definition}
As we have mentioned before, strong causality is equivalent to the $S$-property in spacetimes (see \cite[Prop. 5.1]{szabados1}). In the context of Lorentzian pre-length spaces it is possible to prove a similar result under the assumptions of strong causality and causal path-connectedness for $X$ (with essentially the same proof). However, for Lorentzian metric spaces, a slightly more general approach is needed. In the following result the notions of curve and causal path-connectedness are weakened into separability and Hausdorffness of the topology.

\begin{proposition}
\label{prop:main:2}
Let $(X,\sigma,\tau)$ be a separable and strongly causal Lorentzian metric space. Then, $I^-(x)\sim_S I^+(x)$ for any $x\in X$. If, in addition, $\sigma$ is Hausdorff, then $X$ satisfies the S-property.
\end{proposition}
\begin{proof}
Take $x\in X$ arbitraty, and consider the proper indecomposable sets $I^-(x)$ and $I^+(x)$. Since $X$ is strongly causal, necessarily $I^{\pm}(x)\neq \emptyset$; in fact, given $x\in X$ and an open set $U\ni x$, there exists $p^{\pm}\in I^{\pm}(x)$, and thus, $x\in I(p^-,p^+)$. From Proposition \ref{thm:main:4}, $(X,\sigma,\tau)$ is chronologically dense. So, there exist future and past chains $\left\{ x_n^{\pm} \right\}_n$ converging to $x$ (and so, with $I^{\pm}(x)=I^{\pm}[\sigma^{\pm}]$).

First, recall that $I^+(x)\subset \uparrow I^-(x)$ and $I^-(x)\subset \downarrow I^+(x)$. Assume by contradiction that they are not S-related, and so, there exists $F\neq \emptyset$ with $I^+(x)\subsetneq F\subset \uparrow I^-(x)$ (the case of $P$ breaking the maximality of $I^-(x)$ in $\downarrow F$ is analogous). Take a past chain $\{y_n\}$ defining $F$. Since $F\neq I^+(x)$, the sequence $\{y_n\}$ cannot converge to $x$. So, there exists an open set $(x\in) U$ such that $y_n\not\in U$ for infinitely many $n$. Then, the strong causality provides $p,q\in X$ such that $x\in I(p,q)\subset U$. Since $x\in I(p,q)$, there exists $x^-_{n_0},x^+_{n_0}$ with $p\ll x_{n_0}^-\ll x_{n_0}^+\ll q$. In particular, $x\in I(x^-_{n},x^+_{n})\subset U$ for all $n\geq n_0$. However, since $I^+(x)\subsetneq F\subset \uparrow I^-(x)$, we can find for any $n$ big enough another $m$ such that $x^-_n\ll y_m\ll x^+_n$, i.e. $y_m\in I(x^-_n,x^+_n)(\subset U)$. This contradicts that $y_m$ does not belong to $U$. In conclusion, $I^-(x)\sim_S I^+(x)$.

\smallskip

For the last assertion, suppose that $(X,\sigma,\tau)$ is Hausdorff, and assume by contradiction the existence of $x\neq y$ with $I^-(x)\sim_S I^+(y)$. As mentioned in Remark \ref{rem:top} \ref{rem:top:2}, if $(X,\sigma,\tau)$ is a strongly causal and chronologically dense Lorentzian metric space, the topology $\sigma$ coincides with the chronological topology. However, from Lemma \ref{thm:main:3}, the sequence $\{x_n^-\}$ converges to both, $(I^-(x),I^+(x))\neq (I^-(x),I^+(y))$. This enters in contradiction to the Hausdorff character of the topology $\sigma$. 
\end{proof}

Now we are ready to define the \emph{c-completion} as a set of points:
\begin{definition}\label{def:ccompLPLS}
The {\em c-completion} of a chronologically dense Lorentzian metric space $(X,\sigma,\tau)$ satisfying the S-property is the set 
$$\barX := \{(P,F) \in \hat{X}_{\emptyset} \times \check{X}_{\emptyset}: P \sim_{S} F\}.$$
The {\em c-boundary} of $X$ is defined as $\partial X := \barX \setminus \bfi (X)$, where $\bfi : X \to \barX$ is the inclusion map given by $\bfi (x) := (I^{-}(x), I^{+} (x))$. 
\end{definition}

\begin{remark}
	The c-completion $\barX$ (as a set of points) coincides with the (firstly introduced) \emph{Marolf-Ross completion} (see \cite{MR}). This construction is characterized as the maximum admissibl completion, since it is formed by all the pairs whose components are S-related between them (see \cite{onthefinal}). Even though this completion sometimes becomes redundant, it is the only admissible one that is canonically determined (see \cite[Remark 3.2]{onthefinal}).
\end{remark}

The c-completion $\overline{X}$ is endowed with the so-called \emph{chronological topology}, that is, the coarsest topology among the \emph{admissible} ones (see \cite{onthefinal}). The chronological topology is defined in terms of the following limit operator (see also \cite{flores:revisited,floresharris,onthefinal}).
\begin{definition} \label{def:chrtop}
	Let $\sigma = \{(P_n,F_n)\}$ be a sequence of pairs in $\overline{X}$. A pair $(P,F) \in \overline{X}$ satisfies
	$$(P,F) \in L(\sigma) \Leftrightarrow P \in \hat{L}(P_n) \text{ if } P\neq \emptyset, \text{ and } F \in \check{L}(F_n) \text{ if } F \neq \emptyset,$$
	where
	\begin{align*}
		P\in \hat{L}(P_n) &\Leftrightarrow \left\lbrace \begin{array}{c}
			P\subset LI (P_n) \\
			P \text{ is a maximal IP into } LS(P_n),
		\end{array} \right. \\
		F\in \check{L}(F_n) &\Leftrightarrow \left\lbrace \begin{array}{c}
			F\subset LI (F_n) \\
			F \text{ is a maximal IF into } LS (F_n),
		\end{array} \right.
	\end{align*}
and being $LI$ and $LS$ the usual inferior and superior limit operators in set theory.
\end{definition}
\noindent The \emph{chronological topology $\sigma_{chr}$} for $\overline{X}$ is defined as the sequential topology associated to the limit operator $L$; i.e., a subset $C$ is {\em closed} for $\sigma_{chr}$ if and only if ${L}(\sigma) \subset C$ for any sequence $\sigma \subset C$.

\begin{remark}\label{rem:top}
Some observations are in order.
\begin{enumerate}[label=(\roman*)]
\item 
If $\sigma$ is a sequential topology defined by a limit operator $L$, then
  \[x\in L(\left\{ x_n \right\})\quad \Longrightarrow\quad \left\{ x_{n}  \right\} \hbox{ converges to $x$ with the topology $\sigma$}.
  \]
However, the implication to the left is only true if $L$ is \emph{of first order} (see \cite{onthefinal}).

\item \label{rem:top:2} A natural question to ask is under which conditions the chronological topology $\sigma_{chr}$ coincides with the space topology. In spacetimes, they coincide if the spacetime is strongly causal (in this case they also coincide with the Alexandrov topology). The same result follows in this setting, since the proof of that result is valid for any chronologically dense chronological set (see for instance \cite[Remark 6]{onthefinal}).
\item \label{rem:top:3} There exists another sequential topology that can be used in the context of c-completions, proposed initially by Beem in \cite{Beem_1977} (see also \cite{costa:hausdorff}). That topology, called \emph{closed limit topology} or CLT for short, is also defined by a limit operator: a point $P$ in the future completion belongs to the limit of a sequence $\left\{ P_n \right\}$ if $P=LI(P_n)=LS(P_n)$ (compare with the definition of $\hat{L}$).

  In contrast with the chronological topology, the CLT topology is always metrizable, and they both coincide iff the chronological topology is Hausdorff (see \cite[Theorem 5.3]{costa:hausdorff}). There are however some limitations for the CLT topology. On the one hand, the spacetime need to be \emph{causally continuous} to ensure that the CLT topology coincides with the manifold one (see \cite[Proposition 4.1]{costa:addendum20}). On the other hand, the CLT topology is only defined for partial completions. To extend the CLT topology to the total completion, the future and past completions need to be either completely disconnected (globally hyperbolic spacetimes) or trivially identified from the perspective of the CLT (e.g. globally hyperbolic spacetimes with timelike boundary); see \cite{costa:hausdorff}.
\end{enumerate}
\end{remark}

Finally, we define a Lorentzian distance on $\overline{X}$ by using the Lorentzian distance $\tau$ of $(X,\sigma,\tau)$. This definition strongly rely on the 
chronological density of $X$:

\begin{definition} Let $(X,\sigma,\tau)$ be a chronologically dense Lorentzian metric space satisfying the S-property. We denote by
$\bartau : \barX \times \barX \to [0,\infty]$ the map given by
$$\overline{\tau} ((P,F), (P',F')) := \lim_{n \to \infty} \tau (q_n , p'_n),$$
for any past-directed chain $\{q_n\}$ generating $F$, and any future-directed chain $\{p'_n\}$ generating $P'$. If either $F=\emptyset$ or $P'=\emptyset$, we take $\overline{\tau}((P,F),(P',F'))=0$.
\end{definition}
In order to show that $\bar\tau$ is a Lorentzian distance on $\overline{X}$ we need to verify a number of preliminary properties.
\begin{lemma}\label{justbartau}
	The function $\bartau : \barX \times \barX \to [0,\infty]$ is well defined, that is, the limit always exists and it does not depend on the choice of chains.
\end{lemma}

\begin{proof}
  First, let us show that the limit in the definition of $\overline{\tau}$ exists.  Let $(P,F), (P',F') \in \overline{X}$, $\{q_n\}$ a past-directed chain generating $F$ and $\{p'_n\}$ a future-directed chain generating $P'$.
We distinguish two cases: either $\tau(q_n,p'_n)=0$ for all $n$ (and so, the sequence $\tau(q_n,p'_n)$ is constantly equal to $0$), or there is some $n_0$ with $\tau(q_{n_0},p'_{n_0})>0$. In this second case, and recalling that $\tau(q_{n+1},q_n),\tau(p'_n,p'_{n+1})>0$, the reverse triangle inequality ensures that
\[ \tau (q_{n+1} , p'_{n+1})\geq \tau (q_{n+1} , q_{n})+\tau (q_{n} , p'_{n})+\tau (p'_{n} , p'_{n+1})>\tau (q_n, p'_n). \]
Therefore, the sequence $\tau(q_n,p'_n)$ is increasing, and so, convergent in $(0,\infty]$.

	On the other hand, $\overline{\tau}$ does not depend on the choice of the chains generating $F$ and $P'$. In fact, let $\{q_n\}, \{r_n\}$ be past-directed chains generating $F$ and $\{p'_n\}, \{s'_n\}$ future-directed chains generating $P'$. For every $n$ there exists $m$ with 
	\[
	r_m \ll q_n \ll p'_n \ll s'_m\quad\hbox{and}\quad q_m \ll r_n \ll s'_n \ll p'_m 
	\]
	and so,
	$$ \tau( r_m , s'_m) \geq \tau (q_n , p'_n)\quad\hbox{and}\quad \tau( q_m , p'_m) \geq \tau (r_n , s'_n).$$
	 Therefore, both sequences $\tau(q_n,p'_n)$ and $\tau(r_{n},s'_n)$ converge to the same value, showing that $\overline{\tau}$ is independent of that choice.
       \end{proof}
       
Next, we prove that the extended Lorentzian metric $\overline{\tau}$ satisfies the reverse triangle inequality. To this aim, we begin by showing that $\overline{\tau}$ recovers the chronological relation on c-completions (as defined in \cite[Definition 3.6]{onthefinal}). More precisely:

\begin{lemma}\label{def:extendedrelations}
	The following equivalence holds:
	$$\overline{\tau}((P,F),(P',F'))>0 \quad \Leftrightarrow \quad F \cap P' \neq \emptyset\qquad\forall\; (P,F),(P',F') \in \overline{X}.$$
\end{lemma}

\begin{proof} If $F \cap P' \neq \emptyset$, by transivity there exist $q_n \in F$ and $p'_n \in P'$ with $q_n \ll p'_n$, and thus, $\tau (q_n, p'_n) > 0$. As we have seen before, the sequence $\tau(q_n,p'_{n})$ is increasing, so $\overline{\tau} ((P,F),(P',F')) > 0$.
	
	Reciprocally, if $\overline{\tau} ((P,F),(P',F'))>0$ then there exist $q_n\in F$ and $p'_n \in P'$ with $\tau (q_n , p'_n) > 0$, and again by transivity $F\cap P ' \neq \emptyset$.
\end{proof}

Denote by $\overline{\ll}$ the chronological relation associated to $\overline{\tau}$.  We can now prove the following result:

\begin{lemma}[Reverse triangle inequality]\label{reversetriangle}
Let $\overline{X}$ be the c-completion of a Lorentzian metric space $(X,\sigma,\tau)$. If $(P,F) \overline{\ll} (P',F') \overline{\ll} (P'',F'')$ for $(P,F),(P',F'),(P'',F'') \in \overline{X}$, then
\[ \overline{\tau} ((P,F), (P'',F'')) \geq \overline{\tau} ((P,F), (P',F')) + \overline{\tau} ((P',F'), (P'',F'')).\]
\end{lemma}

\begin{proof}
Let $\{q_n\}$ and $\{q'_n\}$ be past-directed chains generating $F$ and $F'$, respectively. Let $\{p'_n\}$ and $\{p''_n\}$ be future-directed chains generating $P'$ and $P''$, respectively. By the reverse triangle inequality for $\tau$ we have
\begin{equation}\label{gy}
\tau (q_n, p''_n) \geq \tau (q_n, p'_n) + \tau (p'_n,q'_n) + \tau (q'_n, p''_n) \geq \tau (q_n, p'_n) + \tau (q'_n, p''_n), 
\end{equation}
where we have used that $p'_n\leq q'_n$ (recall that $P'\sim_S F'$). We conclude by taking limits in (\ref{gy}), since the corresponding sequences are increasing (thus, convergent).
\end{proof}

Finally, we are going to show that $\bartau$ is lower semi-continuous. To this aim, first we need a technical lemma whose proof requires the following characterization of lower semi-continuity \cite[Section 6.2]{bourbaki}:

\begin{proposition}
A real-valued function $f$ on a topological space $X$ is lower semi-continuous if and only if, for each finite real number $k$, the set $f^{-1}(]k,+\infty])$ is open in $X$ or, equivalently, the set $f^{-1} ([-\infty, k])$ is closed in $X$.
\end{proposition}

\begin{lemma}\label{lemma:lsc}
Let $(X,\sigma)$ be a topological space with $\sigma$ being the sequential topology associated to a limit operator $L$. A function $f: X \to [0,+\infty]$ is lower semi-continuous at $x$ if, for every $\epsilon >0$ and any sequence $\{x_n\} \subset X$ with $x \in L(\{x_n\})$, there exists a subsequence $\{x_{n_k}\}$ such that
\[ f(x_{n_k}) > f(x) - \epsilon\qquad\hbox{for all $k$}.\]
\end{lemma}

\begin{proof}
We are going to prove that $C_a := f^{-1} ([0,a])$ is closed for every $a\in \mathbb{R}$; that is, for any sequence $\{x_n\} \subset C_a$ with $x \in L (\{x_n\})$, necessarily $x \in C_a$, that is, $f(x)\leq a$, or equivalently, $f (x) < a + \epsilon$ for any $\epsilon > 0$. Take $\epsilon >0$ arbitrary, and observe that $x_n \in C_a$ for every $n$. Hence, $0 \leq f(x_n) \leq a$. By assumption, there exists a subsequence $\{x_{n_k}\}$ such that
\[ f(x) - \epsilon < f(x_{n_k}) \leq a,\]
So, taking the limit in previous inequality, we deduce that $f (x) < a + \epsilon$, as required.
\end{proof}

\begin{lemma}[Lower semi-continuity]\label{lowersemicont}
The function $\overline{\tau}:\barX \times \barX \to [0,\infty]$ is lower semi-continuous with respect to the chronological topology $\sigma_{chr}$ of $\overline{X}$.
\end{lemma}

\begin{proof}
Consider two pairs $(P,F),(P',F') \in \barX$, and two sequences $\left\{ (P_{k},F_{k}) \right\},\left\{ (P'_{k},F'_{k}) \right\}$ with $(P,F)\in L(\left\{(P_k , F_k) \right\})$ and $(P',F')\in L(\left\{(P'_k , F'_k) \right\})$. If either $F$ or $P'$ is empty then, from definition $0=\overline{\tau}((P,F),(P',F'))\leq \overline{\tau}((P_{k},F_{k}),(P_{k}',F_{k}'))$ for all $k$. So, assume that both $F$ and $P'$ are non-empty. Since $F\subset \mathrm{LI}(\left\{ F_k \right\})$ and $P'\subset \mathrm{LI}(\left\{ P'_k \right\})$, necessarily $F_k\neq \emptyset\neq P'_k$ for $k$ big enough. Consider a past-directed chain $\{q_n\}$ generating $F$, a future-directed chain $\{p'_n\}$ generating $P'$ and, for each $k$, a future-directed chain $\{r'_{n,k}\}$ generating $P'_k$ and a past-directed chain $\left\{ s_{n,k} \right\}$ generating $F_k$. {Since the sequence $\tau (q_n , p'_n)$ is strictly increasing whenever $q_n \ll p'_n$, then $$\bartau((P,F),(P',F')) = \lim_n \tau (q_n,p'_n) = \sup \{\tau (q_n, p'_n):n\in \mathbb{N}\}.$$ So, given $\epsilon>0$, there exists $N \in \mathbb{N}$ such that 
\[\tau (q_N, p'_N) > \bartau ((P,F),(P',F')) - \epsilon.\] 
Taking into account that $F\subset \mathrm{LI}(F_k)$ and $P' \subset \mathrm{LI} (P'_k)$, we get, for $n$ large enough and for all $k$ (up to a finite number) that,
\[ s_{n,k} \ll q_N \ll p'_N \ll r'_{n,k}. \]
Then, the reverse triangle inequality for $\tau$ ensures that
\[ \tau (s_{n,k} , r'_{n,k}) > \tau (q_N, p'_N)  > \bartau ((P,F),(P',F')) - \epsilon,\]
and, by taking limits in $n$, implies
\[ \overline{\tau} ((P_{k},F_{k}), (P'_k, F'_k)) \geq \overline{\tau} ((P,F), (P',F')- \epsilon.\]
By Lemma \ref{lemma:lsc} we conclude that $\bartau$ is lower semi-continuous.}
\end{proof}

\noindent Summarizing:

\begin{theorem}
If $(X,\sigma,\tau)$ is a separable and chronologically dense Lorentzian metric space that satisfies the S-property, then $(\overline{X},\overline{\sigma}:=\sigma_{chr},\overline{\tau})$ is also a Lorentzian metric space.
\end{theorem}
\noindent As a direct consequence:
\begin{corollary}
  The c-completion of any strongly causal Lorentzian length space is defined and it has a structure of Lorentzian metric space.
\end{corollary}





\subsection{Remarkable properties of the c-completion}

In this subsection we are going to prove a couple of results showing that, under mild hypothesis, the c-completion of a Lorentzian metric space is actually a completion (in the sense of Definition \ref{defcompl}) that possesses a number of additional satisfactory properties. The first result can be formulted as follows: 

\begin{theorem}
\label{thm:main:2}
Let $(X,\sigma,\tau)$ be a separable and strongly causal Lorentzian metric space, with $\sigma$ a Hausdorff topology. Then, the c-completion $(\overline{X},\overline{\sigma}:=\sigma_{chr}, \overline{\tau})$ is a $T_1$ completion with closed boundary $\partial X=\overline{X}\setminus X$. Moreover, if $X$ is connected then $\overline{X}$ is also connected.
\end{theorem}
\noindent The proof of this theorem will be broken down into several steps given by different results. Here $(X,\sigma , \tau)$ satisfies the $S$-property, since it is a separable and strongly causal Lorentzian metric space with $\sigma$ Hausdorff (recall Proposition \ref{prop:main:2}). First, we are going to prove that $(\overline{X},\overline{\sigma}:=\sigma_{chr},\overline{\tau})$ is an extension of $(X,\sigma,\tau)$ (Definition \ref{exten:dis}).

\begin{proposition}
The inclusion  $\bfi:(X,\sigma)\rightarrow (\overline{X},\overline{\sigma}:=\sigma_{chr})$ is a topological embedding.
\end{proposition}

\begin{proof}
  The arguments are essentially the same as the ones in the proof of \cite[Theorems 4.2 and 6.3]{flores:revisited}. Since $(X, \sigma, \tau)$ is strongly causal, it is distinguishing and $\bfi$ is injective. Moreover, $\sigma$ coincides with the Alexandrov topology (and also with CEAT), and CEAT coincides with the chronological topology whenever \emph{full spaces}\footnote{Recall that any strongly causal Lorentzian metric space is \emph{full}, \cite[Proposition 3.16]{onthefinal}} are considered (see also \cite[Remark 3.17]{onthefinal}).
\end{proof}

\smallskip

\noindent We have verified that condition (i) of Definition \ref{exten:dis} holds. Next, we are not just going to prove that condition (ii) holds, but also that the c-completion preserves the Lorentzian distances (in the sense of Definition \ref{defcompl} \ref{defcompl1}):

\begin{proposition}\label{prop:tauextiende}
The function $\overline{\tau}$ {\em extends} $\tau$ to $\overline{X}$ via $\bfi$, that is, if $(P,F) = (I^{-}(q) , I^{+} (q))$ and $(P',F') = (I^{-} (p') , I^{+} (p'))$ then 
$$\overline{\tau} ((P,F), (P',F')) = \tau (q,p').$$
\end{proposition}

\begin{proof}
Let $\{q_n\}$, $\{p'_n\}$ be past and future-directed chains generating $F$ and $P'$, resp. Since $F$ and $P'$ are proper, necessarily
$$q_n \searrow q \qquad \text{and} \qquad p'_n \nearrow p.$$
Given $\delta > 0$, the lower semi-continuity of $\tau$ implies that
\[ \tau (q,p') \geq \tau ( q_n , p'_n) > \tau (q,p') - \delta\quad\hbox{for $n$ big enough}.\]
So, taking $\delta>0$ arbitrarily small, we deduce that 
$$\overline{\tau} ((P,F),(P',F')) = \tau (q,p').$$
\end{proof}

\smallskip


Now, we are in conditions to give the proof of the first theorem of this section (compare with \cite[Theorem 3.27]{onthefinal}). {To this goal, first we recall the following technical lemma, whose proof can be easily adapted from the one in \cite[Theorem 4.4]{FHSSeparability} (see also \cite[Proposition 3.4]{flores:revisited})}:

\begin{lemma}
\label{thm:main:3}
{If $(P,F)\in \overline{X}$ satisfies $P=I^-[\sigma]$ (resp. $F=I^+[\sigma]$) for some chain  $\sigma=\left\{ (P_n,F_n) \right\}_n$, then $(P,F)\in L(\left\{ (P_n,F_n)\right\})$.}
\end{lemma}

\medskip

\begin{proof}{\it (of Theorem \ref{thm:main:2}).} We have already proved that $\overline{X}$ is an extension of $X$. So, let us prove that $\overline{X}$ is both, chronologically complete and tight. Consider a future chronological chain $\overline{\varsigma}=\left\{ (P_n,F_n) \right\}\subset \overline{X}$ (the past case will be analogous), and let us show that $\overline{\varsigma}$ has some endpoint in $\overline{X}$. Since $(P_n,F_n)\ll (P_{n+1},F_{n+1})$, necessarily $P_{n+1}\cap F_n\neq \emptyset$ (Lemma \ref{def:extendedrelations}). Then, taking $z_{n}\in P_{n+1}\cap F_n$, it follows that
  \begin{equation}
    \label{eq:2}
    P_n\subset I^-(z_n)\subset P_{n+1}\quad\hbox{for all $n$}.
\end{equation}
Moreover, $\varsigma=\left\{ z_n \right\}$ defines a future chronological sequence with $P=I^-[\varsigma]$ being an IP. Hence, for a given $F$ $S$-related with $P$, it follows that $(P,F)\in \overline{X}$. Even more, from \eqref{eq:2}, necessarily $P=I^-[\overline{\varsigma}]$. From previous lemma, $(P,F)$ is the endpoint of $\overline{\varsigma}$, and so, $\overline{X}$ is chronologically complete.

For the density, recall that, for a given point $(P,F)\in \overline{X}$ with $P\neq \emptyset$ (the case $F\neq \emptyset$ is analogous), there exists a future chain $\sigma\subset X$ with $P=I^-[\sigma]$. Therefore, the same assertion shows that any $(P,F)\in \overline{X}$ is a limit point of a sequence on $X$, and so, $\bfi (X)$ is dense in $\overline{X}$ (hence, $\overline{X}$ is tight).

The closedness of $\partial X$ is a consequence of the fact that $\bfi(X)$ is open in $\overline{X}$. For the $T_1$ property, note that, for any constant sequence $\left\{ (P,F) \right\}_n\subset \overline{X}$, the element $(P,F)$ is the only one in $L(\left\{ (P,F) \right\})$. So, the points are closed on the chronological topology.

Finally, for the last assertion, assume that $\barX$ is not connected. Then there exist non-empty disjoint open subsets $U_1, U_2 \subset \barX$ with $\barX = U_1 \cup U_2$. Since $\bfi (X)$ is dense in $\barX$, the open sets $U_1 \cap \bfi (X)$ and $U_2 \cap \bfi (X)$ are non-empty, violating the connectedness of $X$.
\end{proof}

\begin{remark}
\label{rem:main:1}
Some observations are in order:
\begin{enumerate}[label=(\roman*)]
    \item The c-completion provides a way to construct an extension (as a Lorentzian pre-length space) of a Lorentzian length space using only its causal structure.
    \item {Strongly causal Lorentzian length spaces, that may not be extendible as a chronologically dense Lorentzian length space (see \cite[Theorem 5.3]{GKSinext}), can always be extended in a natural way to a Lorentzian pre-length space.}
\end{enumerate}
\end{remark}

%

Finally, and inspired by \cite[Section 4]{onthefinal}, we conclude this section by exploring in a second theorem the relative placing of the c-boundary $\partial X$ with respect to other extensions $\tilde{X}$ in the context of Lorentzian metric spaces. To this aim, and taking into acount that these constructions are formally independent, we are going to introduce a procedure to relate points of $\partial X$ with points of the boundary $\tilde{\partial} X$ of $\tilde{X}$ (see Figure \ref{def:associated}).    
\begin{figure}
 \centering
 \includegraphics[scale=0.6]{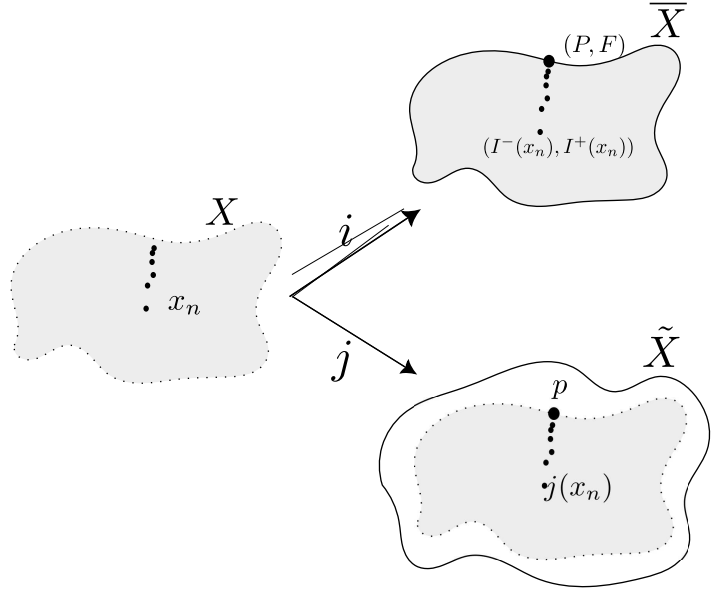}
 \caption{The picture represents two points $p\in \tilde{X}$ and $(P,F)\in \overline{X}$ associated between them according to Definition \ref{def:main:5}. In fact, the sequence $\sigma=\left\{ x_n \right\}\subset X$ satisfies that $\left\{ i(x_n)\equiv (P_n,F_n) \right\}\subset \overline{X}$ converges to $(P,F)$ with the chronological topology and
$\left\{ j(x_n)\right\}\subset \tilde{X}$ converges to $p$ with the topology in $\tilde{X}$.}
 \label{def:associated}
\end{figure}

\begin{definition}
\label{def:main:5}
The elements $p\in \tilde{\partial}X$ and $(P,F)\in \partial X$ are {\em associated between them} if they are both the endpoints of a timelike chain $\left\{ x_n \right\}\subset X$; more precisely, $p$ is the endpoint of $\left\{ j(x_n) \right\}$ and $(P,F)$ is the endpoint of $\left\{ (I^-(x_n),I^+(x_n)) \right\}$, with respect to the corresponding topologies. 
\end{definition}
\noindent The result at hand establishes that the c-boundary $\partial X$ satisfies a sort of ``density property'' between points of X and $\tilde{\partial}X$. 
\begin{theorem}
\label{prop:main:4} Let $(\tilde{X},\tilde{\sigma},\tilde{\tau})$ be an extension of a Lorentzian metric 
space $(X,\sigma,\tau)$. Assume that it is timelike path-connected\footnote{In the sense that any pair of points $x, y\in X$ with $x\ll y$ can be joined by a continuous future directed timelike curve in $X$, i.e. a non-constant continuous map $\gamma : I (\subset \mathbb{R}) \to X$ with the property that $\gamma (s) \ll \gamma (t)$ for all $s < t$.} and preserves the Lorentzian distances. Then, for any pair of points $p\in j(X)$ and $q\in \tilde{X}\setminus j(X)$ with $p\ll q$ (resp. $q\ll p$) there exists some point of the c-boundary associated to some $r\in \tilde{X}$ with $p\ll r\ll q$ (resp. $q\ll r\ll p$).
\end{theorem}
\begin{proof}
  Since $\tilde{X}$ is path-connected, there exists a future timelike curve $\gamma:[a,b]\rightarrow \tilde{X}$ joining $p$ with $q$. Since $\gamma$ is a continuous map and $q\in \tilde{X}\setminus j(X)$, there exists a first value $t_0\in (a,b]$ with $\gamma(t_0)\in \tilde{X}\setminus j(X)=\tilde{\partial} X$. In particular, $\gamma\mid_{[a,t_0)}\rightarrow X$ is an inextensible timelike curve in $X$, and taking any sequence $\left\{ t_n \right\}\subset [a,t_0)$ with $t_n\nearrow t_0$, the sequence $\varsigma=\left\{ \gamma(t_n) \right\}_n$ is an inextensible future timelike chain converging to some $(P,F)\in \partial X$, with $P=I^-[\varsigma]$ (recall Lemma \ref{thm:main:3}). By construction, the points $\gamma(t_0)\in\tilde{\partial}X$ and $(P,F)\in\partial X$ are associated between them, since they are endpoints of the same future timelike chain $\left\{ \gamma(t_n) \right\}_n$. Moreover, $p\ll \gamma(t_0)\ll q$, as required.
\end{proof}

\section{Interplay between the causal structure of $X$ and $\overline{X}$} 
\label{sec:applications}

In this section we establish some interesting relations between the causal structure of the original Lorentzian metric space and its c-completion. The first result is an adaptation to the setting of Lorentzian metric spaces of a well-known result for spacetimes, which characterizes the global hyperbolicity of the original spacetime in terms of the non-existence of c-boundary points whose past and future are both non-empty (\cite[Theorem 3.29]{onthefinal}; see also \cite[Theorem 9.1]{flores:revisited}). 

\begin{proposition} \label{prop:bordeglobalhyp}
	Let $(X,\sigma,\tau)$ be a separable and strongly causal Lorentzian metric space with a compatible causal relation $\leq$. Then,
	 the following properties are equivalent:
	\begin{enumerate}[label=(\roman*)]
		\item $X$ is globally hyperbolic;
		\item if $(P,F) \in \partial X$ then either $P = \emptyset$ or $F = \emptyset$;
		\item if $(P,F) \in \partial X$ then either $\uparrow P = \emptyset$ or $\downarrow F = \emptyset$.
	\end{enumerate}
\end{proposition}

\begin{proof}
	(i) $\Rightarrow$ (ii). By contradiction, assume the existence of some pair $(P,F) \in \partial X$ with $P \neq \emptyset \neq F$. Choose points $p \in P, q \in F$, and a chain $\varsigma \subset X$ generating $P$, and so, converging to $(P,F)$. Then, $\varsigma$ is eventually contained in $I^{+} (p) \cap I^{-} (q)$. Moreover, the terminal character of $P$ ensures that no subsequence of $\varsigma$ converges in $X$. Therefore, $J(p,q) \subset M$ is not compact, in contradiction with the global hyperbolicity of $M$.
	
	(ii) $\Rightarrow$ (iii). Assume the existence of $(P, \emptyset) \in \partial X$ with $\uparrow P \neq \emptyset$. Let $\emptyset \neq F$ be a maximal IF into $\uparrow P$. Notice that $F$ must be terminal; in fact, otherwise, $F = I^{+} (p)$, and any chain $\varsigma$ generating $P$ would satisfy that $\varsigma$ converges to $p$ with the chronological topology on $X$, in contradiction with the terminal character of $P$. Let $\emptyset \neq P \subset P'$ be some maximal IP into $\downarrow F$. Then $P' \sim_S F$, in contradiction with (ii).
	
	
	(iii) $\Rightarrow$ (i). Given $p,q \in X$, we wish to prove that $J(p,q)$ is compact. Assume by contradiction that $J(p,q)$ is not. There is a sequence $\varsigma=\left\{ p_n \right\} \subset J (p,q)$ with no subsequence converging in $X$. By the push-up property,
	\begin{equation*}
		I^{-} (p) \subset I^{-} (p_n), \qquad I^{+} (q) \subset I^{+} (p_n)\quad\hbox{for all $n$.}
	\end{equation*}
	 Moreover, $I^{-} (p) \subset LS (I^{-} (p_n))$. Then, by applying \cite[Theorem 5.11]{floresharris} to $\hat{X}$, there exists a subsequence $\varsigma^{\infty} \subset \varsigma$ and a (necessarily) TIP $P$ such that $\emptyset \neq I^{-} (p) \subset P$ with $P \in \hat{L} (\varsigma^{\infty})$. Furthermore, $P$ provides some pair $(P,F)$ with $ \emptyset \neq P \subset \downarrow F$ and $\uparrow P \supset I^{+}(q) \neq \emptyset$, a contradiction.
\end{proof}

The strong causality of a Lorentzian metric space is not preserved when passing to the c-completion. In fact, this can be easily seen in a globally hyperbolic spacetime $M$, which, according to Proposition \ref{prop:bordeglobalhyp}, possesses a c-boundary formed by pairs $(P,F)$ having some of its components empty. If, for instance, $F=\emptyset$, the chronological future of the element $(P,\emptyset)\in \overline{M}$ is empty. So, $(P,\emptyset)$ is not contained in any chronological diamond of $\overline{M}$, and consequently, it is also not contained in any open set of the Alexandrov topology of $\overline{M}$ (violating the strong causality of the c-completion). 

Nevertheless, it is still possible to prove a result in the positive direction if, rather than aim for strong causality (according to Definition \ref{def:causalityconditions} \ref{esc:4}), we content ourselves with a local convexity property for the c-completion (which is actually closer to the classical formulation of strong causality). More preciely:

	\begin{proposition}
		Let $(\barX,\overline{\sigma}:=\sigma_{chr},\overline{\tau})$ be the c-completion of a full, separable globally hyperbolic and chronologically dense Lorentzian metric space $(X,\sigma,\tau)$ satisfying the S-property. Assume that there are no pairs in $\partial X$ causally related between them with $\leq_{\overline{\tau}}$ (recall the definition given in Section \ref{sec:lorentzianmetricspaces}). For any $(P,F) \in \barX$ and any neighbourhood $U$ of $(P,F)$, there is another neighbourhood $V$ of $(P,F)$ such that, if $(R,Q),(R',Q'),(R'',Q'')\in \overline{X}$ satisfy $(R,Q)\ll (R',Q')\ll (R'',Q'')$ and $(R,Q),(R'',Q'')\in V$ then $(R',Q')\in U$.
	\end{proposition}

\begin{proof}
	Since $X$ is globally hyperbolic, it is also strongly causal. Moreover, $\bfi(X)$ is open in $\overline{X}$. Therefore, the result holds for any point in $\bfi (X)$; in fact, given $x\in X$ and a neighbourhood $U$ of $x$, there are points $p,q\in X$ with $x\in I(p,q)\subset U$, and consequently, we can take $V:= I(p,q)$. So, we can focus on points $(P,F)\in \partial X$.

	By Proposition \ref{prop:bordeglobalhyp}, either $P= \emptyset$ or $F = \emptyset$. Assume that $F = \emptyset$ (the case $P=\emptyset$ is analogous). Let $U$ be a neighbourhood of $(P, \emptyset)$ for the chronological topology and $\left\{ p_n  \right\}$  a future chain generating $P$, i.e. $P = I^{-} [\{ p_n \}]$. Then, $(P,\emptyset) \in I^{+} (\bfi (p_n))$ for all $n$. It suffices to prove that $I^{+} (\bfi (p_n)) \subset U$ for $n$ large enough.
	
	Assume by contradiction the existence of $(R_n,Q_n) \in I^{+} (\bfi (p_n)) \cap U^c$ for infinitely many $n$. Our aim is to prove that $\varsigma=\left\{ (R_n,Q_n) \right\}$ converges to $(P,\emptyset)$, in contradiction with the inclusion of $\varsigma$ into the closed set $U^c$.
	Since $\left\{ p_n \right\}$ is a (future) chain defining $P$, necessarily $P\subset LI(p_n)\subset LI(R_n)\left(\subset LS(R_n) \right)$. If $P$ is not maximal in $LS(R_n)$, there exists $P'$ with $P\subsetneq P'\subset LS(R_n)$. The IP $P'$ must be a TIP, since otherwise $\uparrow P\subset \uparrow P'\neq \emptyset$, violating Proposition \ref{prop:bordeglobalhyp} (iii). The same proposition ensures that $(P',\emptyset)\in \partial X$. Hence, the pairs $(P,\emptyset), (P',\emptyset)\in \partial X$ are causally related, in contradiction with the hypothesis.
	Summarizing, $P$ is maximal in $LS(R_n)$, and consequently, $\varsigma=\left\{ (R_n,Q_n) \right\}$ converges to $(P,\emptyset)$.
\end{proof}

The following result etablishes that, as for the classical regular case of spacetimes, global hyperbolicity implies strong causality under mild hypotheses: 

\begin{proposition}\label{krut}
Any full, separable, chronologically dense and globally hyperbolic Lorentzian metric space $(X,\sigma,\tau)$ satisfying the S-property is strongly causal.
\end{proposition}

\begin{proof}
	Let $x \in X$ and consider a neighborhood $U$ of $x$ in $X$. Let $\{p_n\}$ and  $\{q_n\}$ be a future and a past chain, resp., converging to $x$. For $n$ large enough, we have that $p_n,q_n \in U$. Since $I^{\pm} (\cdot)$ are open (recall Proposition \ref{thm:main:1}), it is enough to prove that $I^{+} (p_n) \cap I^{-} (q_n)$ is inside $U$ for $n$ large enough. By contradiction, assume the existence of $r_n \in (I^{+} (p_n) \cap I^{-} (q_n) ) \setminus U$ for every $n$. Then, $r_n$ belongs to the compact set $J(p_1,q_1)$, and consequently, there must exist a subsequence $\{r_{n_k}\}$ converging to $r$. Since every $r_{n_k}$ is inside $U^{c}$, and $U$ is open, necessarily $r \in U^{c}$. Notice also that $r \in \cap_{n \in \mathbb{N}} J(p_n, q_n)$. Moreover, for every $p \in I^{-} (x)$ there exists $p_n$ with $p \ll p_n \leq r$, which by the push-up property implies $I^- (x) \subset I^- (r)$. Similarly, $I^+ (x) \subset I^+ (r)$. Thus, every $s \in I^+ (r)$ is in the future of all $p_n$ since $p_n \leq r \ll s$, which implies that $I^+ (x) \subsetneq I^+ (r) \subset \uparrow I^- ({p_n}) = \uparrow I^- (x)$, in contradiction with the maximality of $I^+ (x)$ into $\uparrow I^- (x)$, and so, violating the S-property hypothesis. In conclusion, $I^+ (p_n) \cap I^{-} (q_n) \subset U$ for $n$ large enough, and consequently, $(X,\sigma,\tau)$ is strongly causal.
\end{proof}


%

Finally, we are going to characterize the causal simplicity of the original Lorentzian metric space in terms of the c-completion. To this aim, we previously need the following technical result:

\begin{lemma} \label{lema:relcausalJ} Let $x, y$ be two points of a full, separable, chronologically dense Lorentzian metric space $(X,\sigma,\tau)$ satisfying the S-property. Then, $\bfi (x) \leq_{\overline{\tau}} \bfi (y)$ if and only if $y \in \overline{J^{+} (x)}$, $x \in \overline{J ^{-} (y)}$. 
\end{lemma}

\begin{proof}
For the implication to the right, assume that $\bfi (x) \leq_{\overline{\tau}} \bfi (y)$, and thus, $I^{-} (x) \subset I^{-} (y)$ and $I^{+} (y) \subset I^{+} (x)$. Take a past-directed chain $\{y_n\}$ defining $I^+(y)$. Then, $y \in \overline{I^{+} (y)} \subset \overline{I^{+} (x)} \subset \overline{J^{+} (x)}$. Analogously, we deduce that $x \in \overline{J^{-} (y)}$.
	
For the implication to the left, consider $y \in \overline{J^{+} (x)}$ and take $y' \in I^{+} (y)$. Since $I^-(y')$ is open and contains $y$, it intersects $J^+(x)$. Then, the push-up property ensures that $y'\in I^+(x)$, and thus, $I^{+} (y) \subset I^{+}(x)$. By a similar argument we deduce that $I^{-} (x) \subset I^{-} (y)$, and so, $\bfi (x) \leq_{\overline{\tau}} \bfi (y)$. 
\end{proof}

\begin{proposition}\label{prop:causallysimple}
	A full, separable and chronologically dense Lorentzian metric space $(X,\sigma,\tau)$ that satisfies the $S$-property is causally simple if and only if the causal relation $\leq_{\overline{\tau}}$ on $\barX$ (recall the paragraph below Definition \ref{xuxu}) restricted to $\bfi (X)$ coincides with the causal relation $\leq_{\tau}$ on $X$.
\end{proposition}

\begin{proof}
	For the impliction to the right, assume that $X$ is causally simple. Since $\leq_{\tau}$ satisfies the push-up property, then $x \leq_{\tau} y$ implies $\bfi (x) \leq_{\overline{\tau}} \bfi (y)$. For the converse, assume that $\bfi(x)\leq_{\overline{\tau}}\bfi(y)$. From Lemma \ref{lema:relcausalJ}, $x \in \overline{J^{-} (y)} = J^{-} (y)$ and $y \in \overline{J^{-} (x)} = J^{-} (x)$. Hence, $x\leq_{\tau} y$.  
	
	For the implication to the left, assume that $y \in \overline{J^{+} (x)}$ and $x \in \overline{J^{-} (y)}$. By Lemma \ref{lema:relcausalJ}, $\bfi (x) \leq_{\overline{\tau}} \bfi (y)$, and by hypothesis, $x \leq_{\tau} y$, that is, $y \in J^{+} (x)$ and $x \in J^{-} (y)$. Therefore, $J^{+} (x)$ and $J^{-} (y)$ are closed, and consequently, $X$ is causally simple.
\end{proof}

\section*{Acknowledgements}
The first author is partially supported by the project PID2020-116126GB-I00 (funded by MCIN/AEI/10.13039/501100011033)
and the IMAG-María de Maeztu grant CEX2020-001105-M (funded by MCIN/AEI/10.13039/50110001103). The second and third authors are partially supported by the Grants PID2020-118452GBI00 and PID2021-126217NBI00 (Spanish MICINN), resp. These last two authors are also partially supported by the project PY20-01391 (PAIDI 2020, Junta de Andalucía-FEDER), and they would like to acknowledge the kind hospitality of the Erwin Schrödinger International Institute for Mathematics and Physics (ESI), where this research were concluded.

\bibliographystyle{plain}

\end{document}